%% file: main.tex
\documentclass[10pt]{article}

\usepackage{amsmath,amsfonts,amssymb,amsthm}
\usepackage[utf8]{inputenc}
\usepackage{xspace}
\usepackage{graphicx}

\newcommand{\pull}{$\mathcal{PULL}$}
\newcommand{\threestate}{$\textrm{U-Dynamics}$}
\newcommand{\threestateprocess}{$\textrm{U-Process}$}
\newcommand{\pruned}{$\textrm{U-Pruned-Process}$}
\newcommand*{\hmajority}[1]{#1-\textsc{Majority}\xspace}
\newcommand*{\twochoice}{2-\textsc{Choices}\xspace}

\newcommand{\Prob}[2]{\mathbf{P}_{#1} \left( #2 \right)}
\newcommand{\Expec}[2]{\mathbf{E}_{#1} \left[ #2 \right]}
\newcommand{\Ex}[2]{\mathbf{E}_{#1} \left[ #2 \right]}
\newcommand{\cond}{\vert}

\newcommand{\abs}[1]{\vert #1 \vert}
\newcommand{\card}[1]{\vert #1 \vert}

\newcommand{\bigO}{\mathcal{O}}

\newcommand{\bx}{\mathbf{x}}
\newcommand{\x}{\textbf{x}}
\newcommand{\X}{\textbf{X}}
\newcommand{\by}{\mathbf{y}}
\newcommand{\y}{\textbf{y}}
\newcommand{\Y}{\textbf{Y}}
\newcommand{\bz}{\mathbf{z}}

\newcommand{\colA}{\texttt{Alpha}}
\newcommand{\colB}{\texttt{Beta}}

\newtheorem{theorem}{Theorem}[section]

\newtheorem{corollary}{Corollary}

\newtheorem{claim}[theorem]{Claim}
\newtheorem{lemma}[theorem]{Lemma}

\renewcommand{\leq}{\leqslant}
\renewcommand{\geq}{\geqslant}
\renewcommand{\le}{\leqslant}
\renewcommand{\ge}{\geqslant}
\renewcommand{\epsilon}{\varepsilon}

\begin{document}

\title{\textbf{A Tight Analysis of the Parallel Undecided-State Dynamics with
Two Colors}}

\author{
  Andrea Clementi, Francesco Pasquale, Luciano Gual\`a \\
  Universit\`a di Roma ``Tor Vergata''\\
  \texttt{[clementi,pasquale,guala]@mat.uniroma2.it}
  \and
  Emanuele Natale\\
  MPII - Saarbr\"ucken, DE\\
  \texttt{enatale@mpi-inf.mpg.de}
  \and
  Mohsen Ghaffari\\
  ETH Z\"urich\\
  \texttt{ghaffari@inf.ethz.ch}
  \and
  Giacomo Scornavacca\\
  UnivAQ - L'Aquila, IT\\
  \texttt{giacomo.scornavacca@graduate.univaq.it}
}

\maketitle

\begin{abstract}
\input{./abstract-fra.tex}

\end{abstract}


\input{./trunk/intro-main.tex}

\input{./trunk/intro-our.tex}

\input{./trunk/intro-related.tex}

\input{./trunk/preliminaries.tex}

\input{./trunk/overview.tex}

\input{./trunk/betterbound.tex}
\input{./trunk/majority.tex}

\input{./trunk/conclusions.tex}

\bibliographystyle{plain}
\bibliography{ussb}

\onecolumn
\newpage
\appendix
\begin{center}
\Huge{\textbf{Appendix}}
\end{center}

\input{./trunk/apx-inequalities.tex}
\end{document}

%% file: abstract-fra.tex
The \emph{Undecided-State Dynamics} is a well-known protocol for distributed
consensus. We analyze it in the parallel \pull\ communication model on the
complete graph for the \emph{binary} case (every node can either support one of
\emph{two} possible colors, or be in the undecided state). 

An interesting open question is whether this dynamics \emph{always} (i.e.,
starting from an arbitrary initial configuration) reaches consensus
\emph{quickly} (i.e., within a polylogarithmic number of rounds) in a complete
graph with $n$ nodes.  Previous work in this setting only considers initial
color configurations with no undecided nodes and a large \emph{bias} (i.e.,
$\Theta(n)$) towards the majority color.

In this paper we present an \textit{unconditional} analysis of the
Undecided-State Dynamics that answers to the above question in the affirmative.
We prove that, starting from \textit{any} initial configuration, the process
reaches a monochromatic configuration within $\bigO(\log n)$ rounds, with high
probability. This bound turns out to be tight. Our analysis also shows that, if the initial configuration
has bias $\Omega(\sqrt{n\log n})$, then the dynamics converges toward the
initial majority color, with high probability.

%% file: trunk/intro-main.tex
\section{Introduction}
Simple local mechanisms for \emph{Consensus} problems in distributed systems
recently received a lot of attention~\cite{AAE07, AD, DGMSS11, Doty14,
MosselNT14, PVV09}. In one of the basic versions of the consensus problem the
system consists of anonymous entities (\textit{nodes}) each one initially
supporting a \emph{color} out of a finite set of colors $\Sigma$. Nodes run
elementary operations and interact by exchanging messages. A \emph{Consensus
Protocol} is a local procedure that makes the system converge to a
\emph{monochromatic} configuration, where all nodes support the same color. The
consensus has to be \emph{valid}, i.e., the ``winning'' color must be one of
those initially supported by at least one node. The consensus has to be
\textit{stable}, i.e., once the system reaches a monochromatic configuration,
it must remain in that configuration forever, unless some external event takes
place.  Stable consensus is a fundamental building-block that plays an
important role in coordination tasks and self-organizing behavior in population
systems~\cite{cardelli2012cell,C12,Doty14,NB11}.

We study the consensus problem in the \emph{\pull\ communication model}
\cite{CHKM12, DGHILSSST87, KDG03} where, at every round, each active node of a
communication network contacts one neighbor uniformly at random to pull
information. A natural consensus protocol in this model is the
\emph{Undecided-State Dynamics}\footnote{In some previous papers~\cite{PVV09}
on the binary case ($\card{\Sigma} =2$), this protocol has been also called the
\emph{Third-State Dynamics}. We here prefer the term ``undecided'' since it
also holds for the non-binary case and, moreover, the term well captures the
role of this additional state.} (for short, the \emph{\threestate}) in which
the state of a node can be either a color or the \emph{undecided state}. When a
node is activated, it pulls the state of a random neighbor and updates its
state according to the following updating rule (see
Table~\ref{tab:The-update-rule}): If a colored node pulls a different color
from its current one, then it becomes undecided, while in all other cases it
keeps its color; moreover, if the node is in the undecided state then it will
take the state of the pulled neighbor.
 
The \threestate\ has been previously studied in both
\emph{sequential}~\cite{AAE07} and \emph{parallel}~\cite{BecchettiCNPS15}
models: Informally, in the former only one random node is activated at every
round and it updates its state according to the local rule, while in the latter
all nodes are activated at every round and they update their state,
synchronously. 

As for the sequential model\footnote{\cite{AAE07} in fact considers the
\emph{Population-Protocol} model which is, in our specific context, equivalent
to the sequential \pull\ model.}, \cite{AAE07} provides an unconditional
analysis showing (among other results) that the \threestate\ solves the
\emph{binary} consensus problem (i.e., when $\card{\Sigma} = 2$) in the
complete graph with $n$ nodes within $\bigO(n\log n)$ activations (and, thus,
with an $\bigO(\log n)$ work per node), \emph{with high
probability}\footnote{As usual, we say that an event $\mathcal{E}_n$ holds
\emph{w.h.p.} if $\Prob{}{\mathcal{E}_n}\ge 1 - n^{-\Theta(1)}$.}.  We remark
that the stochastic process induced by the parallel dynamics significantly
departs from the one induced by the sequential dynamics. As a simple evidence
of such qualitative differences, observe that, starting from a configuration
with no undecided nodes, in the parallel case the system might end up in the
non-valid configuration where all nodes are undecided (this would happen if,
for example, at the first round every node pulled a node with the other color).
On the other hand, it is easy to see that in the sequential case the process
always ends up in a monochromatic configuration with no undecided nodes, unless
it starts from a configuration with all nodes undecided. The crucial difference
lies in the random number of nodes that may change color at every round: In the
sequential model, this is at most one\footnote{This number actually becomes $2$
if the sequential communication model activates a random edge per round, rather
than one single node~\cite{AAE07}.}, while in the parallel one, \emph{all}
nodes may change state in one round and, for most phases of the process, the
expected number of changes is indeed linear in $n$. The above difference is one
of the main reasons why no general techniques are currently available to extend
any quantitative analysis for the sequential process to the corresponding
parallel one (and vice versa): The parallel process turns out to be a
non-reversible Markov chain having a transition di-graph of very large degree.
The analysis in~\cite{AAE07} strongly uses the fact that only one node can
change state in one round in order to derive a suitable supermartingale
argument to bound the stopping time of the process. It thus fully covers the
case of sequential interaction models, but it is not helpful to understand the
evolution of the \threestate\ process on any interaction model in which the
number of nodes that may change state in one round is not bounded by some
absolute constant.

As for the parallel \pull\ model, while it is easy to verify that the
\threestate\ achieves consensus in the complete graph (with high probability),
the convergence time of this dynamics is still an interesting open issue, even
in the binary case. Indeed, in~\cite{BecchettiCNPS15} the authors   analyze the
\threestate\ in the parallel \pull\ model on the complete graph for any number 
$k = o(n^{1/3})$ of colors.
However, their analysis requires the initial configuration to have a
relatively-large \emph{bias} $s = c_1-c_2$ between the size $c_1$ of the
(unique) initial plurality and the size $c_2$ of the second-largest color. More
in details, in~\cite{BecchettiCNPS15} it is assumed that $c_1 \geq \alpha c_2$,
for some absolute constant $\alpha >1$ and, thus, this condition for the binary
case would result into requiring a very-large initial bias, i.e., $s =
\Theta(n)$. This analysis clearly does not show that the \threestate\
efficiently solves the binary consensus problem, mainly because it does not
manage \emph{balanced} initial configurations.

%% file: trunk/intro-our.tex
\subsubsection*{Our results}  
We prove that, starting from any color 
configuration\footnote{Our analysis  also considers initial configurations 
with undecided nodes.} on the complete graph, the \threestate\ 
reaches a monochromatic configuration (thus consensus) within $\bigO(\log n)$
rounds, with high probability.  This bound is tight since, for some 
(in fact, a large number of) initial configurations, the process requires 
$\Omega(\log n)$ rounds to converge.

Not assuming a large initial bias of the majority color significantly 
complicates the analysis. Indeed, the major technical issues arise from the 
analysis of \textit{balanced} initial configurations where the system ``needs''
to \emph{break symmetry} without having a strong expected drift towards any 
color. Previous analysis of this phase consider either 
\textit{sequential} processes of interacting particles that can be modeled as
\textit{birth-and-death} chains~\cite{AAE07} or parallel processes whose local
rule is fully symmetric w.r.t. the states/colors of the nodes (such as majority
rules)~\cite{becchetti2016stabilizing,DGMSS11}. The \threestate\ process falls 
neither in the former nor in the latter scenario: It works in parallel rounds 
and the role of the undecided nodes makes the local rule not symmetric.
We believe this issue has a \emph{per-se} scientific 
interest since    
  symmetry-breaking phenomena yielded 
by simple and local mechanisms plays a central role in 
 key aspects of  population systems \cite{BK08} and, more generally, in 
the emerging field of natural algorithms \cite{C12}.

Informally speaking, in Section~\ref{sec:break} we deal with almost-balanced
starting configurations. We show that the analysis of this
\textit{symmetry-breaking} phase essentially reduces to the analysis of a 
specific regime where the number $q$ of undecided nodes remains a suitable
constant fraction of $n$ \emph{until} the magnitude of the bias $s$ reaches 
$\Omega(\sqrt{n\log n})$: In other words, during this regime, with very high
probability the system never jumps to almost-balanced configurations having 
either too many or too few undecided nodes. This fact is crucial   
for two main  reasons: Along this regime, (i) the \emph{variance} of the bias $s$ is
large (i.e. $\Theta(n)$) and (ii) whenever the bias $s$ is $\Omega(\sqrt n)$,
its drift turns out to be \emph{exponential} with non-negligible, increasing 
probability (w.r.t. $s$ itself). Then, by devising a coupling to a ``simplified''
pruned process, we can apply (a suitable variant) of a general 
Lemma~\cite{DGMSS11}  that provides a
logarithmic bound on the hitting time of  Markov chains satisfying   Properties
(i) and (ii) above.
 
The symmetry-breaking phase terminates when the \threestateprocess\ reaches 
some configuration having a bias $s = \Omega( \sqrt{n \log n})$. Then (see 
Section~\ref{sec:majority}) we prove that, starting from \emph{any} 
configuration having that bias, the process reaches consensus within 
$\bigO(\log n)$ rounds, with high probability. Even though our analysis of 
this ``majority'' part of the process is based on standard concentration 
arguments, it must cope with some \emph{non-monotone} behavior of the key 
random variables (such as the bias and the number of undecided nodes at the 
next round): Again, this is due to the non-symmetric role played by the
undecided nodes. A good intuition about this ``non-monotone'' process can be 
gained by looking at the mutually-related formulas giving the expectation of 
such key random variables (see Equations~\eqref{exp:a}-\eqref{exp:s}). Our 
refined analysis shows that, during this majority phase, the winning color 
never changes and, thus, the \threestate\ also ensures Plurality Consensus in
logarithmic time whenever the initial bias is $s = \Omega(\sqrt{n \log n})$.

Interestingly enough, we also show that configurations with 
$s = \bigO(\sqrt{n})$ exist so that the system may converge toward the 
minority color with non-negligible probability.

%% file: trunk/intro-related.tex
\subsubsection*{Further motivation and related work} 

\noindent
\textbf{On the \threestate.}
The interest  in the \threestate\ arises in fields beyond the borders of 
Computer Science and it seems to have a key-role in important biological 
processes modeled as so-called chemical reaction 
networks~\cite{cardelli2012cell,Doty14}. 
For such reasons, the convergence time of this dynamics has been analyzed on 
different communication models~\cite{AD12, AAE07, BD13, BTV09, CER14, DGMSS11,
DV12, MertziosNRS14, PVV09}.  \\
As previously mentioned, the \threestate\ has been analyzed in the parallel 
\pull\ model in~\cite{BecchettiCNPS15} and their results concern the evolution
of the process for the multi-color case when there is a significant initial 
bias (as a protocol for plurality consensus).

\noindent
As for the sequential model, the \threestate\ has been introduced and analyzed
in~\cite{AAE07} on the complete graph.
They prove that this dynamics, with high probability, converges to a valid 
consensus within $\bigO(n\log n)$ activations and, moreover, it converges to 
the majority whenever the initial bias is $\omega\left(\sqrt{n \log n}\right)$. 

\noindent  
Still concerning the sequential model, \cite{MertziosNRS14} recently analyzes,
besides other protocols, the \threestate\ in arbitrary graphs where in the initial
configuration each node samples uniformly at random one out of two colors.
In   this  (average-case) setting,
 they prove that the system converges to the initial majority color with higher probability than the initial minority one.
 They also give   results  for special classes of graphs where the minority can win with large probability if the initial configuration is
 chosen in a suitable way.
 Their proof for this result relies on an exponentially-small upper bound on
 the probability that a certain minority can win in the complete graph 
 (see~\cite{MertziosNRS14} for more details).\\
  In  \cite{BD13,BTV09,DV12,PVV09}, the same dynamics for the binary case has been analyzed 
in   other  sequential communication  models. 

\smallskip
\noindent
\textbf{On some other   consensus dynamics.}
Recently, further   simple    consensus protocols  have been  deeply  analyzed in several papers, thus witnessing
the high interest of the  scientific community on  such processes \cite{AAE07,BCNPST14,BGKM16,cardelli2012cell,CER14,CERRS15,DGMSS11,PVV09}.

The parallel \hmajority{3} is a protocol where at every round, each node picks 
the colors
of three random neighbors and updates its color according to the majority rule (taking the first one or a random
one to break   ties). The
authors of~\cite{BCNPST14} assume that the bias is $\Omega({ \min\{
\sqrt{2k}, {(n / \log n)}^{1/6} \} \cdot \sqrt{ n \log n } })$. Under this
assumption, they prove that consensus is reached with high probability in
$\bigO({\min\{ k, {(n / \log n)}^{1/3} \} \cdot\log n})$ rounds,
and that this is tight if $k \leq {(n / \log n)}^{1/4}$. The first  result without
bias~\cite{becchetti2016stabilizing} restricts the number of initial colors to $k =
\bigO({n^{1/3}})$. Under this assumption, they prove that
\hmajority{3} reaches consensus with high probability in
$\bigO({(k^2 {(\log n)}^{1/2}+ k \log n) \cdot (k + \log n)})$
rounds. Very recently, such result has been generalized to the whole range of $k$ in \cite{PetraAL17}. 

In  \cite{DGMSS11} the authors   provide an   analysis 
of the \emph{3-median} rule, in which every node   updates its  value to 
the median of its random sample.
They show that this dynamics   converges to  an almost-agreement configuration (which is 
even a good approximation of the global  median) within 
$\bigO(\log k \cdot \log\log  n + \log n)$ rounds, w.h.p.  
It turns out that, in  the binary case, the median rule is equivalent to    
the \twochoice dynamics, a  
  variant of \hmajority{3},  thus    their result implies that this   
 is   a stabilizing consensus  protocol   with $\bigO (\log n)$ convergence time. 
 As mentioned earlier, our analysis borrows a  hitting-time  bound on general Markov chains from 
 \cite{DGMSS11}.  
 
 Very recently,  \cite{GL17} provides  an optimal 
  bound $\Theta(k \log n)$ for the  \twochoice dynamics on the complete graph even under some 
  dynamic adversary.
In~\cite{CER14, CERRS15}, the authors consider  the \twochoice dynamics  for plurality consensus in the binary case (i.e. $k=2$).  
   For random $d$-regular graphs,~\cite{CER14} proves that all
nodes agree on the majority  color in $\bigO(\log n)$ rounds, provided that
the bias is $\omega( n \cdot \sqrt{1/d + d/n} )$. The same holds for
arbitrary $d$-regular graphs if the bias is $\Omega(\lambda_2 \cdot n)$,
where $\lambda_2$ is the second largest eigenvalue of the transition matrix. In
\cite{CERRS15}, these results are extended to general expander graphs.

%% file: trunk/preliminaries.tex
\section{Preliminaries}\label{sec:preli}
We analyze the parallel version of the dynamics called \threestate\ in the 
(uniform) \pull\ model on the complete graph:
Starting from an initial configuration where every node supports a color, i.e. 
a value from a set $\Sigma$ of $k$ possible colors\footnote{W.l.o.g. we can define $\Sigma = [k]$ where $[k] = \{1,2,\cdots,k\}$}, at every round, each node $u$
pulls the color of a randomly-selected neighbor $v$. If the color of node $v$
differs from its own color, then node $u$ enters in an \textit{undecided}
state (an extra state with no color). 
When a node is in the undecided state and pulls a color, it gets that color. 
Finally, a node that pulls either an undecided node or a node with its 
own color remains in its current state.

\begin{table}[!ht]
\centering
\begin{tabular}{|c||c|c|c|}
\hline 
$u\big\backslash v$ & undecided & color $i$ & color $j$\tabularnewline
\hline 
\hline 
undecided & undecided & $i$ & $j$\tabularnewline
\hline 
$i$ & $i$ & $i$ & undecided\tabularnewline
\hline 
$j$ & $j$ & undecided & $j$\tabularnewline
\hline 
\end{tabular}
\caption{\label{tab:The-update-rule} The update rule of the  \threestate\
where $i,j\in\left[k\right]$ and $i\neq j$.} 
\end{table}

\noindent
In this paper we consider the case in which there are two possible colors (say
color \colA\ and color \colB). Let us name $\mathcal{C}$ the space of all 
possible configurations and observe that, since the graph is complete, 
a configuration $\x \in \mathcal{C}$ is uniquely determined by fixing the
number of \colA-colored nodes and the number of \colB-colored ones, 
say $a(\x)$ and $b(\x)$, respectively.

It is convenient to give names also to two other quantities that will appear 
often in the analysis: The number $q(\x) = n - a(\x) - b(\x)$ of undecided 
nodes and the difference $s(\x) = a(\x) - b(\x)$ called the \textit{bias} of $\x$.
Notice that any two of the quantities 
$a(\x), b(\x), q(\x)$, and $s(\x)$ uniquely determine the configuration.
When it will be clear from the context, we will omit $\x$ and write $a, b, q$, 
and $s$ instead of $a(\x), b(\x), q(\x)$, and $s(\x)$.

Observe that the \threestate\ defines a finite-state Markov chain 
$\{\X_t \}_{t \geq 0}$ with state space $\mathcal{C}$ and three absorbing states,
namely, $q = n$, $a = n$, and $b = n$. We call \emph{\threestateprocess}\ the random 
process obtained by applying the \threestate\ starting at a given state.
Once we fix the  configuration $\x$ at round
$t$ of the process, i.e. $\X_t = \x$, we use the capital letters $A,B,Q$, and $S$ to 
refer to the random variables $a(\X_{t+1}), b(\X_{t+1}), q(\X_{t+1})$, $s(\X_{t+1})$. 

From the definition of \threestate\ it is easy to calculate the following 
expected values (see also Section~3 in \cite{BecchettiCNPS15}):
\begin{align}
\Ex{}{A \, \vert \, \X_t = \x} & = a \left(\frac{a + 2q}{n}\right) \, ,
\label{exp:a} \\
\Ex{}{Q \, \vert \, \X_t = \x} & = \frac{q^2 + 2ab}{n} \, ,
\label{exp:q} \\
\Ex{}{S \, \vert \, \X_t = \x} & = \frac{a(a + 2q)}{n} - \frac{b(b + 2q)}{n} 
= s\left(1 + \frac{q}{n}\right) \, .
\label{exp:s}
\end{align}

\subsection{The expected evolution of the \threestateprocess} \label{ssec:expover}
Equations~\eqref{exp:a}-\eqref{exp:s} can be used to have a preliminary 
intuitive idea on the expected evolution of the \threestateprocess.
From~\eqref{exp:s} it follows that the bias $s$ increases exponentially, in
expectation, as long as the number $q$ of undecided nodes is a constant 
fraction of $n$ (say, $q \geqslant \delta n$, for some positive constant 
$\delta$). 
By rewriting~\eqref{exp:q} in terms of $q$ and $s$ we have that
\begin{align}
\Ex{}{Q \, \vert \, \X_t = \x} =& \frac{q^2 + 2ab}{n} 
= \frac{2q^2 + (n-q)^2 - s^2}{2n}\nonumber\\
\geqslant& \frac{n}{3} - \frac{s^2}{2n},\label{eq:explbq}
\end{align}
where in the inequality we used the fact that the minimum of $2q^2 + (n-q)^2$
is achieved at $q = \frac n3$ and its value is $ \frac 23 n^2$. From~\eqref{eq:explbq} it 
thus follows that, as long as the magnitude of the bias is smaller than 
a constant fraction of $n$ (say $s < \frac 23 n$), the expected number of undecided
nodes will be larger than a constant fraction of $n$ at the next round (say, 
$\Expec{}{Q \, \vert \, \X_t = \x} \geqslant \frac n9$). 

When the magnitude of the bias $s$ reaches $\frac 23 n$, it is easy to see that the 
expected number of nodes with the \textit{minority} color decreases 
exponentially. Indeed, suppose w.l.o.g. that \colB\ is the minority color and 
rewrite~\eqref{exp:a} for $B$ and in terms of $b$ and $s$. We get
\begin{equation}\label{eq:expdecrb}
\Expec{}{B \, \vert \, \X_t = \x} = b \left(\frac{b + 2q}{n}\right)
= b \left( 1-\frac{2s+3b-n}{n} \right) \, .
\end{equation}
Hence, when $s > \frac 23 n$ we have that 
$\Expec{}{B \, \vert \, \X_t = \x} \leqslant \frac 23 b$.

The above sketch of the analysis \textit{in expectation} would suggest that
the process should end up in a monochromatic configuration within 
$\mathcal{O}(\log n)$ rounds. Indeed, in Theorem~\ref{thm:mainmajority} we 
prove that this is what happens with high probability (w.h.p., from now on) 
when the process starts from a configuration that already has some bias, namely 
$s = \Omega(\sqrt{n \log n})$.

When the process starts from a configuration with a smaller bias, the analysis
\textit{in expectation} looses its predictive power. As an extreme example, 
observe that when $a = b = \frac n3$ the system is ``in equilibrium'' according to
\eqref{exp:a}-\eqref{exp:s}. However, the equilibrium is ``unstable'' and the
symmetry is broken by the \textit{variance} of the process (as long as 
$s = o(\sqrt n)$) and by the increasing drift towards majority (as soon as $s > \sqrt n$). 
As mentioned in the
Introduction, the analysis of this \textit{symmetry-breaking} phase is the key 
technical contribution of the paper and it will be described in 
Section~\ref{sec:break}. This analysis will show that, starting from any initial
configuration, the system reaches a configuration where the magnitude of the 
bias is $\Omega(\sqrt{n \log n})$ within $\mathcal{O}(\log n)$ rounds, w.h.p.

%% file: trunk/overview.tex
\section{Main results and the digraph of  the \threestateprocess'
phases}\label{sec:overview}
As informally discussed in the introduction, we prove  the two following
results characterizing the evolution of the \threestate\ on the synchronous
\pull\ model in the complete graph.

\begin{theorem}[Consensus]\label{thm:mainselfstabilizing}
Let  the \threestateprocess\ start from any configuration in $\mathcal{C}$.
Then the process converges to a (valid) monochromatic configuration   within
$\bigO(\log n)$ rounds, w.h.p.  Furthermore, if the initial configuration has
at least one colored node (i.e. $q \leq n-1$), then the process converges to a
configuration such that $\abs{s} = n$, w.h.p.
\end{theorem}

\begin{theorem}[Plurality consensus]\label{thm:mainmajority}
Let $\gamma$ be any positive constant. Assume that the \threestateprocess\
starts at any biased configuration such that $\abs{s} \geq \gamma\sqrt{n \log
n}$ and assume w.l.o.g. the majority color is \colA. Then the process converges
to the monochromatic configuration with $a = n$ within $\bigO(\log n)$ rounds,
w.h.p. Furthermore, the result is almost tight in a twofold sense: (i) An
initial configuration exists, with $\abs{s} = \Omega(\sqrt{n \log n})$, such
that the process requires $\Omega(\log n)$ rounds to converge w.h.p. and (ii)
there is an initial configuration with $\abs{s} = \Theta(\sqrt{n})$ such that
the process converges to the minority color with constant probability. 
\end{theorem}

\begin{figure}[tb]
\begin{center}
\includegraphics[scale=0.30]{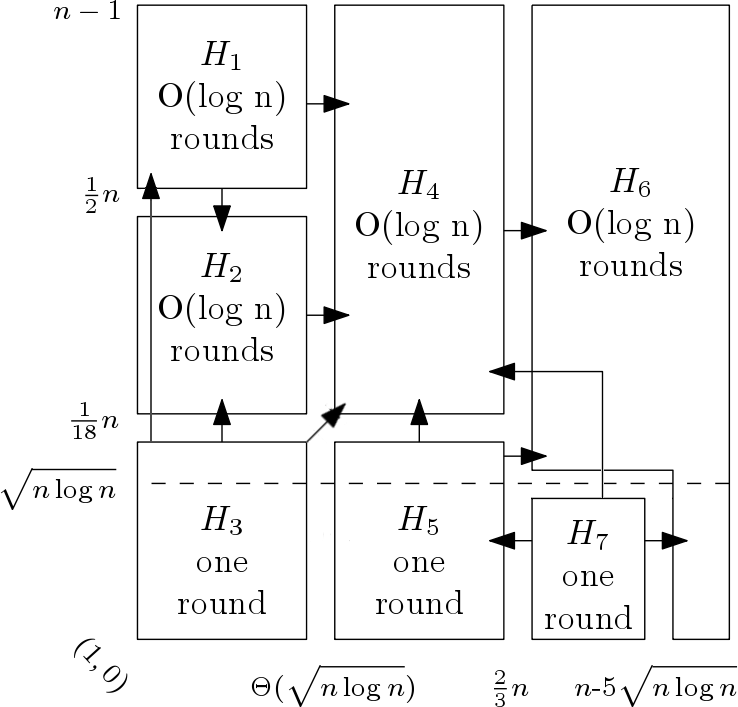}
\end{center}
\caption{$\{H_1,\dots,H_7\}$ is the considered partitioning of the
configuration space $\mathcal{C}$. On the x axis we represent the bias $s$, on
the y axis the number of undecided nodes $q$.  Missing arrows are transitions
that have negligible probabilities.} \label{fig:bigpic}
\end{figure}

\noindent
\textbf{Outline of the two proofs.} The two theorems above are consequences of
our refined  analysis\footnote{We remark that our   analysis  focuses on
asymptotic bounds and it does not definitely optimize the corresponding
constants: However, using technicalities and loosing readability, all such
constants can be largely improved.} of  the evolution of the
\threestateprocess.  The analysis   is organized into a set of possible process
phases,    each of them  is defined by specific ranges of parameters $q$ and
$s$. A high-level description of this structure  is shown  in Fig.\
\ref{fig:bigpic} where every  rectangular region represents a subset of
configurations with specific ranges of $s$ and $q$ and it is associated to a
specific phase.  In details, let $\gamma$ be any positive constant, then  the
regions are defined as follows:
$H_1$ is  the set of configurations such that $s \leq \gamma\sqrt{n \log n}$
and $q \geq \frac{1}{2}n$;   $H_2$ is the set of configurations such that $s
\leq \gamma\sqrt{n \log n}$ and $\frac{1}{18}n \leq q \leq \frac{1}{2}n$; $H_3$
is  the set of configurations such that $s \leq \gamma\sqrt{n \log n}$ and $q
\leq \frac{1}{18}n$. $H_4$ is the set of configurations such that
$\gamma\sqrt{n \log n} \leq s \leq \frac{2}{3}n$ and $q \geq \frac{1}{18}n$; $
H_5$ is the set of configurations such that $\gamma\sqrt{n \log n} \leq s \leq
\frac{2}{3}n$ and $q \leq \frac{1}{18}n$; $H_7$ is the set of configurations
such that $ 	\frac{2}{3}n \leq s \leq n - 5\sqrt{n \log n}$ and $q \leq
\sqrt{n \log n}$.  $H_6$ is the set of configurations such that  $s \geq
\frac{2}{3}n$ minus $H_7$.
     
For each region, Fig.  \ref{fig:bigpic} specifies our  upper bound on the exit
time of the corresponding phase, while black arrows represent the phase
transitions which may happen with non-negligible probability. 

As a first,  important remark, we point out that the scheme of Fig.\
\ref{fig:bigpic} can be seen as a directed acyclic graph $G$, having a single
sink $H_6$ and which is reachable from any other region.  We also remark that,
starting from certain configurations, a monochromatic state may be reached via
different paths in $G$.  This departs from previous analysis of consensus
processes \cite{BecchettiCNPS15,BCNPST14,DGMSS11} in which the phase transition
graph is essentially a path.
    
We now outline the proofs of the  two main results of this paper.
    
\smallskip \noindent
\underline{Outline of the Proof of Theorem \ref{thm:mainmajority}.}  Consider
an initial configuration $\x$ such that $s(\x)\ge \gamma\sqrt{n \log n}$, for
some positive constant $\gamma$, and assume w.l.o.g. that the majority color in
$\x$ is \colA. In Section \ref{sec:majority}, we first show (see Lemma
\ref{lemma:ageofundecideds}) that if the process lies in $H_4$ the bias grows
exponentially fast and thus the process enters in $H_6$ within $\bigO(\log n)$
rounds. Then we prove Lemma \ref{lemma:convergence}, which states that,
starting at any configuration in $H_6$, the process ends in the monochromatic
configuration where $a = n$ within $\bigO(\log n)$ rounds, w.h.p. Next, we show
that, starting from any configuration in $H_5$, the process falls into  $H_4$
or $H_6$ in one round (Lemma \ref{lemma:birthofundecidedsII}) and that,
starting from any configuration in $H_7$, the process falls into $H_4, H_5$ or
$H_6$ in one round (Lemma \ref{lemma:birthofundecidedsIII}).  As for the
tightness of the result stated in the second part of the theorem, we have that
the   lower bound (Claim (i)) on the convergence time is an immediate
consequence of Claim (ii) of Lemma \ref{lemma:ageofundecideds}, while the
second claim, concerning the lower bound on the initial bias, is proved in
Claim \ref{claim:apxminbias}.

\smallskip \noindent    
\underline{Outline of the Proof of Theorem \ref{thm:mainselfstabilizing}.}
We first observe  that the configuration where     all nodes are undecided
(i.e. $q = n$) is an absorbing state of the \threestateprocess\ and thus, for
this initial  configuration, Theorem  \ref{thm:mainselfstabilizing}  trivially
holds.  In Section \ref{sec:break}, we will show that, starting from any
\textit{balanced} configuration, i.e. with $s = o(\sqrt{n \log n}$), the
\threestateprocess\ ``breaks symmetry'' reaching a configuration $\y$ with
$s(\y) = \Omega(\sqrt{n \log n})$ within $\bigO(\log n)$ rounds, w.h.p. Then,
the thesis easily follows by applying Theorem \ref{thm:mainmajority} with
initial configuration $\y$.  As for the  symmetry-breaking phase, in
Lemma~\ref{lemma:starters} we prove that, if the process starts from  a
configuration in $H_1$ or $H_3$ (see Figure~\ref{fig:bigpic}), then after
$\bigO(\log n)$ rounds either the bias  between the two colors becomes
$\Omega(\sqrt{n \log n})$ or the system   reaches some    configuration in
$H_2$, w.h.p.  In Lemma~\ref{lemma:symmetrybreaking} we then prove that, if the
process is in a configuration in $H_2$, then the bias $s$ will  become
$\Omega(\sqrt{n \log n})$ within $\bigO(\log n)$ rounds, w.h.p.

%% file: trunk/betterbound.tex
\section{Symmetry breaking} \label{sec:break}

In this section we show that, starting from any  (almost-) balanced
configuration, i.e. those with $s = o(\sqrt{n \log n}$), the
\threestateprocess\ ``breaks symmetry'' reaching a configuration with $s =
\Omega(\sqrt{n \log n})$ within $\bigO(\log n)$ rounds, w.h.p.  This part of
our analysis is organized as follows.

In Lemma~\ref{lemma:starters} we prove that, if the process starts at a
configuration in $H_1$ or $H_3$ (see Figure~\ref{fig:bigpic}), i.e., when the
number of undecided nodes is either smaller than $n/18$ or larger than $n/2$,
then, after $\bigO(\log n)$ rounds, either the bias between the two colors
already gets magnitude $\Omega(\sqrt{n \log n})$ or the system reaches  some
configuration in $H_2$ (i.e., a configuration where the number of undecided
nodes is between $n/18$ and $n/2$). In Lemma~\ref{lemma:symmetrybreaking} we
then prove that, if the process is in a configuration in $H_2$, then the bias
between the two colors will get magnitude $\Omega(\sqrt{n \log n})$ within
$\bigO(\log n)$ rounds, w.h.p.

Lemma~\ref{lemma:starters} is a simple consequence of the following three
claims. Claims~\ref{claim:q_lower_bound} and~\ref{claim:a_increase} follow from
Chernoff bound applied to~\eqref{eq:explbq} and~\eqref{exp:a}, respectively. 

\begin{claim}\label{claim:q_lower_bound}
Let $\x \in \mathcal{C}$ be any configuration with $\vert s(\x) \vert \leq
(2/3)n$. Then, at the next round, the  number of undecided nodes of the
\threestateprocess\ is $Q \geq n/18$, w.h.p.
\end{claim}
\begin{proof}
From \eqref{eq:explbq} we get
\[
\Expec{}{Q \,\vert\, \X_t = \x} 
\geq \frac{n}{3} - \frac{s^2}{2n} 
\geq \frac{n}{3} - \frac{2}{9}n 
= \frac{n}{9} \,.
\]
By applying the additive form of Chernoff Bound (see~\eqref{eq:cbafleq} in
Appendix~\ref{apx:cbaf}) to the random variable $Q$ we easily get the claim,
i.e., 
\begin{align*} 
\Prob{}{Q \leq \frac{n}{18}} &= \Prob{}{Q \leq \frac{n}{9} - \frac{n}{18}} \\ 
& \leq \Prob{}{Q \leq \Expec{}{Q} - \frac{1}{18}n} \\ 
& \leq e^{-2n^2/18^2n} \\
& = e^{-\Theta(n)}\,.
\end{align*} 
\end{proof}

\begin{claim}\label{claim:a_increase} 
Let $\x \in \mathcal{C}$ be any configuration with $q(\x) \geq n/2$ and $a(\x) 
\geq \log n$. Then, at the next round , the  number of $\colA$-colored nodes of the 
\threestateprocess\ is $A \geq (1 + 1/8) a(\x)$, w.h.p.
\end{claim}
\begin{proof}
Observe that from~\eqref{exp:a} we get 
\begin{align}
\Expec{}{A \,\vert\, \X_t = \x} &= \left( \frac{a + 2q}{n}\right) a \nonumber\\
& \geq \left( \frac{\frac{n - q}{2} + 2q}{n} \right) a\label{align5}\\
& = \left( \frac{1}{2} + \frac{3q}{2n} \right) a \nonumber 
\geq \left( \frac{1}{2} + \frac{3}{4} \right) a \nonumber 
= \frac{5}{4} a,\nonumber
\end{align}
where in~\eqref{align5} we used that $a = \frac{n - q + s}{2} \geq \frac{n -
q}{2}$. By applying the multiplicative form of Chernoff Bound (see
(\ref{eq:cbmfleq}) in Appendix~\ref{apx:cbmf}) with $\delta = 1/10$ we obtain
\begin{align*}
\Prob{}{A \leq \frac{9}{8}} 
& =\Prob{}{A \leq \frac{5}{4} a \, \left( 1-\frac{1}{10} \right)}\\ 
& \leq e^{- \frac{5}{4} a \, \frac{1}{100}/2} \\ 
& \leq e^{- \frac{5}{4} (\log n) \frac{1}{200}} = \frac{1}{n^{\Theta(1)}}.
\end{align*}
where in the second inequality we used hypothesis $a \geqslant \log n$.
Thus, we have that $A > (9/8) a$ w.h.p. 
\end{proof}

The next claim is a consequence of fact that, when the number of colored nodes
is very small, the \threestateprocess\ behaves essentially like a \textit{pull}
process.

\begin{claim}\label{claim:information_spreading} 
Starting at any configuration $\x \in \mathcal{C}$ with $1 \leq a(\x) + b(\x)
< 2 \log n$, the \threestateprocess\ reaches a configuration $\X'$ with $a(\X')
+ b(\X') \geq 2 \log n$ within $\bigO(\log n)$ rounds, w.h.p.
\end{claim}

\begin{proof}
Starting from configuration $\x$, we consider, at every   round $t \geq 1$,   the random 
variable counting the number of \textit{colored 
nodes} $a(\X_t )+ b(\X_t)$. 
Observe that, as long as  $1 \leq a(\X_t )+ b(\X_t) < 2 \log n$, the probability that in one 
round an \colA-colored node picks a \colB-colored node (or vice versa) is 
less than $\frac{(2 \log n)^2}{n}$. Hence, by applying the union bound for 
$\bigO(\log n)$ rounds, we get that the probability that this ``bad'' event 
happens in one of such rounds  is negligible.
Then, we can proceed the analysis  on $a(\X_t )+ b(\X_t)$  by induction and, at every round $t$,
  we assume   the  bad event is not happened so far (so every  colored node  keeps its color).
Now, at   round $t+1$, we only consider those nodes that were undecided at round $t$: each of them becomes colored
 iff it picks a colored node. So, discarding the difference between colors, the process over the 
undecided nodes turns out to be a standard \emph{rumor-spreading} process via the well-known
\pull\ mechanism (a colored node is in fact an informed node).
The claim then follows by observing that this spreading process is known to 
inform at least $2 \log n$ nodes within $\bigO(\log n)$ rounds, w.h.p. (see 
for instance~\cite{KSSV00}).
\end{proof}

\begin{lemma}[Phases $H_1$ and $H_3$: Starters]\label{lemma:starters}
~\begin{itemize}
\item Starting from any configuration $\x \in H_3$, the \threestateprocess\ reaches
a configuration $\X' \in (H_1 \cup H_2 \cup H_4)$ in one round, w.h.p.
\item Starting from any configuration $\x \in H_1$, the \threestateprocess\ reaches 
a configuration $\X' \in (H_2 \cup H_4)$ within $\bigO(\log n)$ rounds, w.h.p.
\end{itemize}
\end{lemma}

\begin{proof}
As for the first statement, from Claim~\ref{claim:q_lower_bound} it follows that the number of undecided nodes at the next round is above $\frac{n}{18}$, hence $\X' \in (H_1 \cup H_2 \cup H_4 \cup H_6)$, w.h.p. However, observe that the process cannot jump directly from $H_3$ to $H_6$ because of a combinatorial reason: from any configuration $(a,b,q)$ with $a\geq b$ the bias at the next round cannot exceed $a+q$ (notice that in a single round a colored node can only keep its color or became undecided).
As for the second statement, let $\bx$ be a configuration in $H_1$. 
If $a(\bx) \geqslant \log n$ then, by repeatedly applying Claim~\ref{claim:a_increase}
we have that either at some point the number of undecided nodes decreases below 
$n/2$ and thus the system reaches a configuration in $H_2$ or, since the number of
$\colA$-colored nodes increases exponentially w.h.p., the bias $s$ increases up to 
$\sqrt{n \log n}$ within $\mathcal{O}(\log n)$ rounds w.h.p., and thus the system 
reaches a configuration in $H_4$.
If $a(\bx) \leqslant \log n$ then, from Claim~\ref{claim:information_spreading} it
follows that within $\mathcal{O}(\log n)$ rounds the system is in a configuration $\mathbf{X}'$ 
with $a(\mathbf{X}') \geqslant \log n$, w.h.p. At this point either we are in the previous case 
or the number of undecided nodes is already $q(\mathbf{X}') \leqslant n/2$.
Notice that Claim~\ref{claim:q_lower_bound} guarantees that during this phase 
the process does not jump into $H_3$ or $H_5$, w.h.p.
\end{proof}
\smallskip

If the system lies in a configuration of $H_2$, we need  more complex probabilistic arguments
 to  prove that the bias between the two colors reaches $\Omega(\sqrt{n \log n})$
within $\bigO(\log n)$ rounds w.h.p.  

We will make use of the following  bound on the hitting time of any Markov chain having suitable
drift properties. This result is a variant of Claim 2.9 in \cite{DGMSS11}.
We provide here an independent proof for our variant since we were not able to find a published proof for the previous one.

\begin{lemma}\label{lemma:symmetrygeneric}
Let $\{X_{t}\}_{t\in \mathbb{N}}$ be a Markov Chain with finite state space $\Omega$ and let 
$f:\Omega\mapsto[0,n]$ be a function that maps states to integer values. 
Let $c_3$ be any positive constant
and let $m = c_3\sqrt{n}\log n$ be a target value.
Assume the following properties hold:
\begin{enumerate}
\item For any positive constant $h$, a positive constant $c_1 < 1$ exists such that for any $x \in \Omega : f(x) < m$,
\[
\Prob{}{f(X_{t+1}) < h\sqrt{n} \cond X_{t} = x} < c_1\,,
\]

\item Two positive constants $\epsilon, c_2$ exist such that for any $x \in \Omega: h\sqrt{n} \leq f(x) < m$,
\[
\Prob{}{f(X_{t+1}) < (1+\epsilon)f(X_{t})\cond X_{t} = x} < e^{-c_2f(x)^2/n}\,.
\]
\end{enumerate}
Then the process reaches a state $x$ such that $f(x) \geq m$ within 
$\bigO(\log n)$ rounds, w.h.p.
\end{lemma}

\begin{proof}

We first define a set of hitting times $T = \{\tau(i)\}_{i \in \mathbb{N}}$ where $$\tau(i) = \inf_{t \in \mathbb{N}} \{ t: t > \tau(i-1) , f(X_t) \geq h\sqrt{n} \}$$
setting $\tau(0) = 0$. By Hypothesis (1), for every $i \in \mathbb{N}$, the expectation of $\tau(i)$ is finite.

Then we define the following stochastic process which is a subsequence of $\{X_{t}\}_{t\in \mathbb{N}}$:
$$\{R_i\}_{i \in \mathbb{N}} = \{X_{\tau(i)} \}_{i \in \mathbb{N}}.$$

Observe that $\{R_i\}_{i \in \mathbb{N}}$ is still a Markov Chain. Indeed, let $\{x_1, \dots, x_{i-1}\}$ a set of states in $\Omega$:

\begin{align*}
    & \Prob{}{R_i = x \vert R_{i-1} = x_{i-1} \wedge \dots \wedge R_{1} = x_1}\\
    &=  \Prob{}{X_{\tau(i)} = x \vert X_{\tau(i-1)} = x_{i-1} \wedge \dots \wedge X_{\tau(1)} = x_1}\\
    &= \sum_{t(i) \wedge \dots \wedge t(0) \in \mathbb{N}} \Prob{}{X_{t(i)} = x \vert X_{t(i-1)} = x_{i-1} \wedge \dots \wedge X_{t(1)} = x_1}\\ 
    &\cdot \Prob{}{\tau(i) = t(i) \wedge \tau(i-1) = t(i-1) \wedge \dots \wedge \tau(1) = t(1)}\\
    &= \sum_{t(i) \wedge \dots \wedge t(0) \in \mathbb{N}} \Prob{}{X_{t(i)} = x \vert X_{t(i-1)} = x_{i-1}}\\ 
    &\cdot \Prob{}{\tau(i) = t(i) \wedge \tau(i-1) = t(i-1) \wedge \dots \wedge \tau(1) = t(1)}\\
    &= \Prob{}{X_{\tau(i)} = x \vert X_{\tau(i-1)} = x_{i-1}}\\
    &= \Prob{}{R_i = x \vert R_{i-1} = x_{i-1}}.
\end{align*}

By definition the state space of $R$ is $\{x \in \Omega: f(x) \geq h\sqrt{n}\}$. Moreover Hypothesis (2) still holds for this new Markov Chain. Indeed:

\begin{align*}
    &\Prob{}{f(R_{i+1}) < (1+\epsilon) f(R_{i}) \vert R_i = x)}\\
    &= 1 - \Prob{}{f(R_{i+1}) \geq (1+\epsilon) f(R_{i}) \vert R_i = x)}\\
    &= 1 - \Prob{}{f(X_{\tau(i+1)}) \geq (1+\epsilon) f(X_{\tau(i)}) \vert X_{\tau(i)} = x)}\\
    &\leq 1 - \Prob{}{f(X_{\tau(i+1)}) \geq (1+\epsilon) f(X_{\tau(i)}) \wedge \tau(i+1) = \tau(i)+1 \vert X_{\tau(i)} = x)}\\
    &= 1 - \Prob{}{f(X_{\tau(i)+1}) \geq (1+\epsilon) f(X_{\tau(i)}) \vert X_{\tau(i)} = x)}\\
    &= 1 - \Prob{}{f(X_{t+1}\geq (1+\epsilon) f(X_{t}) \vert X_{t} = x)}\\
    &< e^{-c_2f(x)^2/n}.
\end{align*}

These two properties are sufficient to study the number of rounds required by the new Markov Chain $\{R_i\}_{i \in \mathbb{N}}$ to reach the target value $m$.
Indeed, by defining the random variable $Z_i = \frac{f(R_i)}{\sqrt{n}}$ and considering the following 
  ``potential'' function, $Y_i = \exp(\frac{m}{\sqrt{n}} - Z_i)$ we can compute its expectation at the next round as follows. Let us fix any state $x \in \Omega$ such that $h\sqrt{n} \leq f(x) < m$ and define $z = \frac{f(x)}{\sqrt{n}}$ and $y = \exp(\frac{m}{\sqrt{n}} - z)$. We get:
\begin{align}
\Ex{}{Y_{i+1} \vert R_i = x} &\leq \Prob{}{f(R_{i+1}) < (1+\epsilon)f(x)} e^{m/\sqrt{n}}\nonumber\\
&+ \Prob{}{f(R_{i+1}) \geq (1+\epsilon)f(x)} e^{m/\sqrt{n} - (1 + \epsilon)z}\nonumber\\
\mbox{ (from Hypothesis  (2)) } &\leq e^{-c_2 z^2} \cdot  e^{m/\sqrt{n}} + 1 \cdot e^{m/\sqrt{n} - (1 + \epsilon)z}\nonumber\\
&= e^{m/\sqrt{n} - c_2 z^2} + e^{m/\sqrt{n} - z - \epsilon z}\nonumber\\
&= e^{m/\sqrt{n} - z} (e^{z - c_2 z^2} + e^{-\epsilon z})\nonumber\\
&\leq e^{m/\sqrt{n} - z}(e^{-2} + e^{-2})\label{eq:smallh}\\
&<\frac{e^{m/\sqrt{n} - z}}{e}\nonumber\\
&<\frac{y}{e}\nonumber,
\end{align}
where in~\eqref{eq:smallh} we used that $z$ is always at least $h$ and thanks to Hypothesis (1) we can choose a sufficiently large $h$.

By applying the Markov inequality and iterating the above bound, we get: 
$$\Prob{}{Y_{i} > 1} \leq \frac{\Ex{}{Y_{i}}}{1} \leq \frac{\Ex{}{Y_{{i}-1}}}{e}\leq \cdots \leq \frac{\Ex{}{Y_0}}{e^{\tau_R}} \leq \frac{e^{m/\sqrt{n}}}{e^{i}}.$$
We observe that if $Y_{i} \leq 1$ then $R_{i} \geq m$, thus by setting ${i} = m/\sqrt{n} + \log n = (c_3 + 1)\log n$, we get:

\begin{equation}
\Prob{}{R_{(c_3 + 1)\log n} < m} = \Prob{}{Y_{(c_3 + 1)\log n} > 1} < \frac{1}{n}.\label{eq:markovR}
\end{equation}

Our next goal is to give an upper bound on the hitting time $\tau_{(c_3 + 1)\log n}$. 
Note that the event "$\tau_{(c_3 + 1)\log n} > c_4\log n$" holds if and only if the number of rounds such that $f(X_t) \geq h\sqrt{n}$ (before round $c_4\log n$) is less than $(c_3 + 1)\log n$.
Thanks to Hypothesis (1), at each round $t$ there is at least probability $1-c_1$ that $f(X_{t}) \geq h\sqrt{n}$. This implies that, for any positive constant $c_4$, the probability $\Prob{}{\tau_{(c_3 + 1)\log n} > c_4\log n}$ is bounded by the probability that, within $c_4\log n$ independent Bernoulli trials, we get less then $(c_3 + 1)\log n$ successes, where the success probability is at least $1-c_1$. We can thus choose a sufficiently large $c_4$ and apply the multiplicative form of the Chernoff bound (see (\ref{eq:cbmfgeq}) in Appendix \ref{apx:cbmf}) and obtain:

\begin{equation}
\Prob{}{\tau_{(c_3 + 1)\log n} > c_4\log n} < \frac{1}{n}.\label{eq:tauR}
\end{equation}

We are now ready to prove the Lemma using Inequalities \eqref{eq:markovR} and \eqref{eq:tauR}, indeed:

\begin{align*}
    \Prob{}{X_{c_4 \log n} \geq m} 
    &> \Prob{}{R_{(c_3 + 1)\log n} \geq m \wedge \tau_{(c_3 + 1)\log n} \leq c_4 \log n}\\
    &= 1 - \Prob{}{R_{(c_3 + 1)\log n} < m  \vee \tau_{(c_3 + 1)\log n} > c_4 \log n}\\
    &\geq 1 - \Prob{}{R_{(c_3 + 1)\log n} < m} + \Prob{}{\tau_{(c_3 + 1)\log n} > c_4 \log n}\\
    &> 1 - \frac{2}{n}.
\end{align*}

Hence, choosing a suitable big $c_4$, we have shown that in $c_4 \log n$ rounds the process reaches the target value $m$, w.h.p.

\end{proof}

\smallskip

The basic idea would be to apply the above lemma to the \threestateprocess\ with $f(X_t) = s(X_t)$ in order to get an upper bound on the number of rounds needed to reach a configuration having bias $\Omega(\sqrt{n \log n})$. To this aim, we first show that, for any configuration in $H_2$, Properties 1 and 2 in Lemma \ref{lemma:symmetrygeneric} are satisfied.

\begin{claim}\label{unstable+goodbias} Let $\x \in \mathcal{C}$ be any configuration such that $\frac{n}{18}\leq q(\x) \leq \frac{n}{2}$ and $\vert s(\x) \vert < c_4 \sqrt{n} \log n$ for any positive constant $c_4$, then it holds: 

\begin{enumerate}
\item for any constant $h > 0$ a constant $c_1 < 1$ exists such that
$$\Prob{}{\vert S \vert < h\sqrt{n}\,\vert\, \X_t = \x} < c_1,$$

\item two positive constants $c_2,\epsilon$ exist such that 
$$\Prob{}{\vert S \vert \geq (1+\epsilon)s\,\vert\, \X_t = \x} 
\geq 1 - e^{-c_2s^2/n}.$$
\end{enumerate}
\end{claim}
\begin{proof}
As for the first item, let $\x$ and $\x_0$ two states such that $\vert s(\x) \vert <  h\sqrt{n}$, $\vert s(\x_0) \vert = 0$ and $q(\x) = q(\x_0)$. By a simple domination argument we get that: 
$$\Prob{}{\vert S \vert < h\sqrt{n} \,\vert\, \X_t = \x} \leq \Prob{}{\vert S \vert < h\sqrt{n} \,\vert\, \X_t = \x_0}.$$
Thus we can consider only the case where the bias is zero and this implies that $a = b$.

We define $A^q,B^q,Q^q$ the random variables counting the nodes that were undecided in the configuration $\x_0$ and that, in the next round, 
get    colored with  \colA, \colB, or undecided, respectively. Similarly $A^a$ ($B^b$) counts the nodes that support
   color \colA\ (\colB) in the configuration $\x_0$ and that, in the next round, still support the same color.

Since it is impossible that a node supporting a color  gets the other color in the next round, it holds that $A = A^q + A^a$ and $B = B^q + B^b$. Moreover, observe  that, among these random variables, only $A^q$ and $B^q$ are mutually dependent. Thus it holds  
\begin{align*}
&\Prob{}{\vert S \vert \geq h\sqrt{n}} > \Prob{}{A \geq B + h\sqrt{n}}\\
&= \Prob{}{A^q + A^a \geq B^q + B^b + h\sqrt{n}}\\
&\geq \Prob{}{A^q \geq B^q + h\sqrt{n} \, , \, A^a \geq aq/n \, , \, B^b \leq aq/n}\\
&= \Prob{}{A^q - B^q \geq h\sqrt{n}} \cdot \Prob{}{A^a - aq/n \geq 0} \cdot \Prob{}{B^b - aq/n \leq 0}
\end{align*}

Note that both the random variables $A^a - aq/n$ and $B^b - aq/n$ are binomial distribution with expectation $0$ (recall that $a=b$). Thus with a simple application of the Berry-Esseen Theorem we can approximate, up to a arbitrary small constant $\epsilon$, both the random variables with a normal distribution with mean zero. Then we get that $\Prob{}{A^a - aq/n \geq 0} \cdot \Prob{}{B^b - aq/n \leq 0} \geq (1/2 - \epsilon)^2$.

As for the random variable $A^q - B^q$, note that conditioned to the event $Q^q = k$ it is a sum of $q-k$ Rademacher random variables. Note that $\Ex{}{Q^q} = q^2/n$, then by using that $q \leq \frac{n}{2}$ we get that $\Ex{}{Q^q} \leq q/2$. By an application of the multiplicative Chernoff Bound and by the fact that $q \geq \frac{n}{18}$  it holds that $q-Q^q = \Theta(n)$ w.h.p. This implies that the variance of $A^q - B^q$ is $\Theta(n)$ w.h.p.

Thus with a simple application of the Berry-Esseen Theorem we can approximate $A^q - B^q$ up to a arbitrarily small constant, with a normal distribution with mean zero and variance $\Theta(n)$ w.h.p. We can conclude:

\begin{align*}
&\Prob{}{A^q - B^q \geq h\sqrt{n}} \cdot \Prob{}{A^a - aq/n \geq 0} \cdot \Prob{}{B^b - aq/n \leq 0}\\
&\geq\Prob{}{A^q - B^q \geq h\sqrt{n}\vert q-Q^q = \Theta(n)} \cdot \Prob{}{q-Q^q = \Theta(n)}\\
&\cdot \Prob{}{A^a - aq/n \geq 0} \cdot \Prob{}{B^b - aq/n \leq 0} \geq c_1
\end{align*}

for a suitable small constant $c_1$.

As for the second item, w.l.o.g. we assume that $a > b$. From the additive form of the Chernoff bound (see Appendix~\ref{apx:cbaf}) 
it follows that
\[
\Prob{}{A < \Ex{}{A} - \frac{1}{72}s} < e^{-2s^2/72^2n}\,
\]
and
\[
\Prob{}{B > \Ex{}{B} + \frac{1}{72}s} < e^{-2s^2/72^2n}\,.
\]
Thus:
\begin{align*}
S &= A - B > \Ex{}{A} - \frac{1}{72}s - \Ex{}{B} - \frac{1}{72}s\\&= \Ex{}{A - B} - \frac{1}{36}s = \Ex{}{S} - \frac{1}{18}s\\ 
&= \left(1 + \frac{q}{n}\right)s - \frac{1}{36}s 
= \left(1 + \frac{1}{18} - \frac{1}{36}\right)s\\
&= \left( 1 + \frac{1}{36} \right)s\,.
\end{align*}
Hence, the second item is obtained setting $\epsilon = \frac{1}{36}$ and a
$c_2 < \frac{2}{72^2}$. 
\end{proof}

\smallskip
 It is important to observe that the above claim ensures Properties 1 and 2 of 
Lemma~\ref{lemma:symmetrygeneric} whenever $\frac{1}{18}n\leq q\leq \frac{1}{2}n$.
Unfortunately, Lemma~\ref{lemma:symmetrygeneric} requires such properties to hold 
for \emph{any} \mbox{(almost-)}balanced configuration: If $q = n-o(n)$, Property 1 does 
not hold, while Property 2 is not satisfied if $q=o(n)$. In order to manage 
this issue, in Subsection~\ref{ssec:pruned}, we define a \emph{pruned} process,
a variant of \threestateprocess\, where it is possible to apply 
Lemma~\ref{lemma:symmetrygeneric}. Then, in Subsection~\ref{ssec:back} we show 
a coupling between the \threestateprocess\ and the pruned one.

\subsection{The pruned process}\label{ssec:pruned}
The helpful, key point is that, starting from any configuration $\x \in H_2$, the probability that the process goes in one of those ``bad'' configurations  with $q<\frac{1}{18}n$ or $q \geq \frac{1}{2}n$ is negligible (see 
Claim~\ref{claim:q_bounded}). Thus, intuitively speaking, all the configurations 
\emph{actually visited} by the process before leaving $H_2$ do satisfy
 Lemma~\ref{lemma:symmetrygeneric}. In order to make
this intuitive argument rigorous, in what follows,  we   define a suitably
\textit{pruned} process by removing from $H_2$ all the \textit{unwanted} 
transitions. 

Let $\bar{s} \in [n]$ and $\bz(\bar{s})$ be the configuration such that 
$s(\bz(\bar{s}))=\bar{s}$ and $q(\bz(\bar{s})) = \frac{1}{2}n$. Let 
$p_{\bx,\by}$ be the probability of a transition from the configuration $\bx$ to the configuration $\by$ in the \threestateprocess. We define a new stochastic process: The \pruned. The \pruned\ behaves exactly like the original process but every transition from a configuration $\bx \in H_2$ to a configuration $\by$ such that $q(\by) < \frac{1}{18}n$ or $q(\by) > \frac{1}{2}n$ now have probability $p_{\bx,\by}' = 0$. Moreover, for any $\bar{s} \in [n]$, starting from any configuration $\bx \in H_2$, the probability of reaching the configuration $\bz(\bar{s})$ is: 
\begin{equation*} p_{\bx,\bz(\bar{s})}' = p_{\bx,\bz(\bar{s})} + \sum_{\by:\left(q(\by) < \frac{1}{18}n \vee q(\by) > \frac{1}{2}n\right)\bigwedge s(\by) = \bar{s}}p_{\bx,\by.}\end{equation*}

\noindent Finally, all the other transition probabilities remain the same.

Observe that, since the \pruned\ is defined in such a way it has exactly the 
same marginal probability of the original process with respect to the random 
variable $s$, then Claim \ref{unstable+goodbias} holds for the \pruned\ as well.
Thus, we can choose constants $h,c_1,c_2,\epsilon$ such that the two properties 
of Lemma~\ref{lemma:symmetrygeneric} are satisfied.  

\begin{corollary} \label{prun:ubound}
 Starting from any configuration  $\x \in H_2$, the \\  \pruned\   reaches 
a configuration $\X' \in H_4$ within $\bigO(\log n)$ rounds, w.h.p.  
\end{corollary}
  
\subsection{Back to the original process.}\label{ssec:back}
 The  definition of the \pruned\ suggests a natural  coupling between the original process and the pruned one: If the two processes are in different states then they act independently, while, if they are in the same configuration $\bx$, they move together unless the \threestateprocess\ goes in a configuration $\by$ such that $q(\by) < \frac{1}{18}n$ or $q(\by) > \frac{1}{2}n$. In that case the \pruned\ goes in $\bz(s(\by))$. 
 
Using  this   coupling, we first show that, if the two processes are in the same configuration, the probability that they get separated is negligible. Then, we show that the $H_2$ exit time of the pruned procedure stochastically dominates 
the $H_2$ exit time of the original process.
 
\begin{claim}\label{claim:q_bounded} 
For every configuration $\x \in H_2$, the probability that the number of 
undecided nodes in the next round of the \threestateprocess\ is not between 
$n/18$ and $n/2$ is
$$\Prob{}{q(\X_{t+1}) \notin \left[\frac{n}{18},\, \frac{n}{2}\right] \,\vert\, \X_t = \x} 
\leqslant e^{-\Theta(n)}.$$
\end{claim}
\begin{proof}
The lower bound directly follows from Claim~\ref{claim:q_lower_bound}. 
In order to show that $q(\X_{t+1}) \leq n/2$ with probability 
exponentially close to $1$, observe that from \eqref{eq:explbq} we have
\begin{align*} 
\Expec{}{q(\X_{t+1}) \, \vert \, \X_t = \x} 
& = \frac{2q^2 + (n-q)^2 - s^2}{2n} \\
& \leq \frac{2q^2 + (n-q)^2}{2n}\,,
\end{align*}
and for $n/18 \leq q \leq n/2$ the maximum of $2q^2 + (n-q)^2$ is obtained at 
$q = n/18$. Hence,
\begin{align*}
\frac{2q^2 + (n-q)^2}{2n} 
& \leq \frac{\frac{2}{18^2}n^2 + \left( n-\frac{1}{18}n \right)^2}{2n} \\ 
&= \left( \frac{1}{2} - c \right)n\,.
\end{align*}
for a constant $c > 0$. By using the additive form of the Chernoff bound (see (\ref{eq:cbafgeq}) in 
Appendix~\ref{apx:cbaf}) with $\lambda = c \cdot n$ and $\mu = 
\Expec{}{q(\X_{t+1}) \, \vert \, \X_t = \x} \leqslant n/2 - \lambda$, we obtain
\begin{align*}
\Prob{}{q(\X_{t+1}) \geq \frac{1}{2}n \,\vert\, \X_t = \x} 
& \leq \Prob{}{q(\X_{t+1}) \geq \mu + \lambda \,\vert\, \X_t = \x}\\
& \leq e^{-2 \cdot c^2 n^2/n} = e^{-\Theta(n)}.
\end{align*}
\end{proof}

\begin{lemma}[Phase $H_2$]\label{lemma:symmetrybreaking}
Starting from any configuration $\x \in H_2$, the \threestateprocess\ reaches 
a configuration $\X' \in H_4$ within $\bigO(\log n)$ rounds, w.h.p.
\end{lemma}

\begin{proof} Let $\{\X_t\}$ and $\{\Y_t\}$ be the original process and the pruned one, respectively. Let $\x 
\in H_2$, note that if $\X_t = \Y_t = \x$ then
\[
\Y_{t+1} = \left\{
\begin{array}{cl}
\X_{t+1} & \mbox{ if } \X_{t+1} \in H_2 \\[2mm]
\bz(s(\mathbf{X}_{t+1})) & \mbox{ otherwise }
\end{array}
\right.
\]
Let  $\tau = \inf \{t \in \mathbb{N} \,:\, |s(\X_t)| 
\geqslant \sqrt{n \log n} \}$ and let 
$
 \tau^{*}  = \inf \{t \in \mathbb{N} \,:\, |s(\Y_t)| 
\geqslant \sqrt{n \log n} \}.
$
For any $\x \in H_2$ and any round $t$ we define $\rho_x^t$ the event $\{\X_t\}$ and $\{\Y_t\}$ \textit{separated} at round $t+1$, i.e. $\rho_x^t = (\X_t = \Y_t = \x) \bigwedge (\X_{t+1} \neq \Y_{t+1})$. Observe that, if the two coupled processes start in the same configuration 
$\x_0 \in H_2$ and $\tau > c \log n$, then either $\tau^* > c \log n$ 
as well, or a round $t \leqslant c \log n$ exists such that, for some $\x \in H_2$ the event $\rho_x^t$ occurred. Hence,

\vspace{-3mm}
\begin{small}
\begin{align}\label{eq:couplingbound}
&\Prob{\x_0,\x_0}{\tau > c \log n} \leqslant \nonumber\\
&\leq \Prob{\x_0,\x_0}{\{\tau^* > c \log n\} \cup 
\left\lbrace 
\begin{array}{l} 
\exists t \leq c\log n \\ 
\exists \x \in H_2
\end{array}:
\rho_x^t \right\rbrace }
\nonumber\\
&\leq \Prob{\x_0,\x_0}{\tau^* > c \log n}+ \Prob{\x_0,\x_0}{
\begin{array}{l} \exists t \leq c\log n \\ 
\exists \x \in H_2
\end{array}:\rho_x^t}.
\end{align}
\end{small}

\noindent
As for the first term in~\eqref{eq:couplingbound}, from the analysis of the
pruned process (Corollary~\ref{prun:ubound}) we have that it is upper
bounded by $1/n$. As for the second term, we get that
\begin{align}
\Prob{\x_0,\x_0}{
\begin{array}{l} \exists t \leq c\log n \\ 
\exists \x \in H_2
\end{array}: 
\rho_x^t} &\leq \sum_{t=1}^{c \log n}\Prob{\x_0,\x_0}{ \exists \x \in H_2: 
\rho_x^t}\nonumber\\
&= \sum_{t=1}^{c \log n}\sum_{\x \in H_2}\Prob{\x_0,\x_0}{\rho_x^t}\nonumber\\
&\leq \sum_{t=1}^{c \log n}\frac{n^2}{e^{-\Theta(n)}}\label{cardinality}\\
&\leq \frac{1}{n}\nonumber\,,
\end{align}
where in \eqref{cardinality} we used Claim~\ref{claim:q_bounded} and the fact 
that $\vert H_2 \vert$ is at most all the possible combinations of   parameters
$q$ and $s$.
\end{proof}

%% file: trunk/majority.tex
\section{Convergence to the majority}\label{sec:majority}


In this section we provide the arguments needed to prove our second main 
result, namely Theorem~\ref{thm:mainmajority}, which essentially states that, 
starting from any sufficiently-biased configuration, the \threestateprocess\ 
converges to the monochromatic configuration where all nodes support the initial majority 
color. We recall that the outline of the proof is given in Section~\ref{sec:overview}. 
Here, we formalize the arguments of the provided high-level description. To increase readability, the  
   proofs of the   technical claims are moved to the appendix.

\paragraph*{Phase $H_4$ (the age of the undecideds)} 

We first show, that under some parameter ranges including $H_4$ (and hence when
the number of the undecideds are large enough), the growth of the bias is 
exponential.

\begin{claim}\label{claim:s_increase}
Let $\gamma$ be any positive constant and $\x \in \mathcal{C}$ be any 
configuration such that $s \geq \gamma\sqrt{n \log n}$ and $q \geq 
\frac{1}{18}n$. Then, it holds that $s(1 + \frac{1}{36}) < S < 2s $, w.h.p.
\end{claim}

\begin{proof}
Recall that $S = A - B$. In order to show that $S>s(1 + \frac{1}{36})$ w.h.p., we provide two independent bounds to the values of $A$ and $B$, respectively. We use the additive form of the Chernoff bound ((\ref{eq:cbafleq}) and (\ref{eq:cbafgeq}) in Appendix \ref{apx:cbaf}) with $\lambda = \frac{\gamma\sqrt{n \log n}}{72}$. Hence, we have
\begin{align*}
\Prob{}{A \leq  \Expec{}{A \vert X_t = \x} - \lambda} &\leq e^{-2\lambda^2/n}\\ 
&= e^{-2\gamma^2 \log n/72^2}\\
&= \frac{1}{n^{\Theta(1)}},
\end{align*} 
\noindent and
\begin{align*}
\Prob{}{B \geq  \Expec{}{B \vert X_t = \x} + \lambda} &\leq e^{-2\lambda^2/n}\\ &= e^{-2\gamma^2 \log n/72^2}\\ &= \frac{1}{n^{\Theta(1)}}.
\end{align*} 
\noindent Then w.h.p.
\begin{align*}
S &>  \Expec{}{A \vert X_t = \x} - \lambda - [B \vert X_t = x] - \lambda\\
&= \Expec{}{A - B \vert X_t = \x} - 2 \lambda\\
&= \Expec{}{S \vert X_t = \x} - 2 \lambda\\
&= s(1 + \frac{q}{n}) - \frac{\gamma\sqrt{n \log n}}{36}\\
&\geq s(1 + \frac{q}{n}) - s/36\\
&\geq s(1 + \frac{1}{18} - \frac{1}{36})\\
&=s(1 + \frac{1}{36}).
\end{align*} 
We now show that $S<2s$ w.h.p. using similar arguments as above. Once again, we use the additive form of the Chernoff bound with $\lambda = \frac{\gamma\sqrt{n \log n}}{4}$. We have
\begin{align*}
\Prob{}{A \geq  \Expec{}{A \vert X_t = \x} + \lambda}& \leq e^{-2\lambda^2/n}\\ &= e^{-2\gamma^2 \log n/16}\\ &= \frac{1}{n^{\Theta(1)}},
\end{align*} 
\noindent and
\begin{align*}
\Prob{}{B \leq  \Expec{}{B \vert X_t = \x} - \lambda} &\leq e^{-2\lambda^2/n}\\ &= e^{-2\gamma^2 \log n/16}\\ &= \frac{1}{n^{\Theta(1)}}.
\end{align*} 
\noindent As a consequence, we have that w.h.p.
\begin{align*}
S &<  \Expec{}{A \vert X_t = \x} + \lambda - [B \vert X_t = x] + \lambda\\
&= \Expec{}{A - B \vert X_t = \x} + 2 \lambda\\
&= \Expec{}{S \vert X_t = \x} + 2 \lambda\\
&= s(1 + \frac{q}{n}) + \frac{\gamma}{2}\sqrt{n \log n}\\
&< s(1 + \frac{1}{2}) + \frac{1}{2}s\\
&=2s.
\end{align*}
\end{proof}

\begin{lemma}[Phase $H_4$]\label{lemma:ageofundecideds}
Let $\x \in H_4$ be a configuration with $a>b$. Then, (i) starting from $\x$, 
the \threestateprocess\ reaches a configuration $\X' \in H_6$ with $a>b$ 
within $\bigO(\log n)$ rounds, w.h.p. Moreover, (ii) an initial configuration
$\y \in H_4$ exists such that the \threestateprocess\ stays in $H_4$ for 
$\Omega(\log n)$ rounds, w.h.p.
\end{lemma}
\begin{proof} We iteratively apply Claim \ref{claim:s_increase} and Claim \ref{claim:q_lower_bound} and after $t=\Theta(\log n)$ rounds we have that either there is a round $t'<t$ such that $s(\X_{t'})>\frac{2}{3}n$ or $s(\X_{t}) > (1+1/36)^t s(\x) \geq (1+1/36)^t$. In both cases, the process has reached a configuration $\X'$ such that $s(\X') \geq \frac{2}{3}n$ and $q(\X') \geq \frac{n}{18}$: So $\X'$ belongs to $H_6$. Since each step of the iteration holds w.h.p. and the number of steps is $\bigO(\log n)$, we easily obtain that the result holds w.h.p. by a simple application of the Union Bound.

Concerning the second part of the lemma, consider an initial configuration $\y$ such that $s(\y) = n^{2/3}$. By iteratively applying (the upper bound of) Claim \ref{claim:s_increase} and Claim \ref{claim:q_lower_bound} for $t= \frac{1}{4} \log n$ rounds, we have that $s(X_t) < 2^t s(\y) = 2^t n^{2/3}=n^{1/4}n^{2/3}=o(n)$.
\end{proof}

\paragraph*{Phase $H_6$ (the victory of the majority)}

This is the phase in which a large bias let the nodes converge to the majority color within a logarithmic number of rounds. We first prove that the number of nodes that support the minority color decreases exponentially fast (Claim \ref{claim:b_decrease}) and that the bias is preserved round by round (Claim \ref{claim:s_preserved} and Claim \ref{claim:q_preserved}). Then, when $b \leq 2\sqrt{n \log n}$, the number of undecided nodes starts to decrease exponentially fast as well (Claim \ref{claim:q_decrease}). At the very end, when there are only few nodes (i.e., $\bigO(\sqrt{n \log n})$) that do not still support the majority color, the minority color disappears in few steps and thus the \threestateprocess\ converges to majority within $\bigO(\log n)$ rounds (Claim \ref{claim:a_wins}). 

\begin{claim}\label{claim:b_decrease}
Let $\x \in \mathcal{C}$ be any configuration such that $\vert s \vert \geq \frac{2}{3}n$ and $b \geq \log n$ then it holds that $B \leq b(1 - \frac{1}{9})$, w.h.p.
\end{claim}

\begin{proof}
W.l.o.g. we assume that $a>b$. From (\ref{eq:expdecrb}), since $s\geq \frac{2}{3}n$, we have that
\begin{align*} \Expec{}{B \, \vert \, X_t = \x} &= b \left( 1-\frac{2s+3b-n}{n} \right)\\
&\leq b \left( 1-\frac{2s-n}{n} \right)\\
&\leq b \left( 1-\frac{\frac{4}{3}n-n}{n} \right)\\
&= b \left( 1- \frac{1}{3}\right).
\end{align*}
\noindent
Then we apply the multiplicative form of the Chernoff Bound ((\ref{eq:cbmfgeq}) in Appendix \ref{apx:cbmf}) with $\delta = \frac{1}{3}$, and we obtain
\begin{align*}
\Prob{}{B \geq (1+\delta)\left( 1- \frac{1}{3}\right)b} &\leq e^{-b(1-\frac{1}{3})\delta^2/3}\\  
&\leq e^{-\log n(1-\frac{1}{3})\delta^2/3}\\ &= \frac{1}{n^{\Theta(1)}}.
\end{align*}
\noindent
As a consequence, we have that w.h.p.
\begin{align*}
B &\leq b(1+\delta)\left(1-\frac{1}{3}\right)\\ &= b\left(1+\frac{1}{3}\right)\left(1-\frac{1}{3}\right)\\
&= b \left(1-\frac{1}{9}\right).
\end{align*}
\end{proof}

In order to iteratively apply the above claim we now show that, if there are enough undecided nodes, the bias is preserved round by round until the number of \colB-colored nodes decreases below $2\sqrt{n \log n}$.

\begin{claim}\label{claim:s_preserved}
Let $\x \in \mathcal{C}$ be any configuration such that $\vert s \vert \ge \frac{2}{3}n$ and $q \geq \sqrt{n \log n}$. Then it holds that $S \ge \frac{2}{3}n$, w.h.p.
\end{claim}

\begin{proof}
W.l.o.g. we assume that $a>b$. We recall that $S = A - B$, thus we provide two independent bounds to the values of $A$ and $B$ respectively. We use the additive form of the Chernoff bound (see (\ref{eq:cbafleq}) and (\ref{eq:cbafgeq}) in Appendix \ref{apx:cbaf}) with $\lambda = \epsilon\frac{\sqrt{n \log n}}{2}$. We have
\begin{align*}\Prob{}{A \leq  \Expec{}{A \vert X_t = x} - \lambda}& \leq e^{-2\lambda^2/n}\\ &= e^{-\epsilon^2\log n/2}\\ &= \frac{1}{n^{\Theta(1)}},
\end{align*}
\noindent and
\begin{align*}\Prob{}{B \geq  \Expec{}{B \vert X_t = x} + \lambda}& \leq e^{-2\lambda^2/n}\\ &= e^{-\epsilon^2\log n/2}\\ &= \frac{1}{n^{\Theta(1)}}.
\end{align*}
\noindent Then it holds that, w.h.p.
\begin{align*}
S &\geq  \Expec{}{A \vert X_t = x} - \lambda - \Expec{}{B \vert X_t = x} - \lambda\\
&= \Expec{}{A - B \vert X_t = x} - 2 \lambda\\
&= \Expec{}{S \vert X_t = x} - 2 \lambda\\
&= s(1 + \frac{q}{n}) - \epsilon\sqrt{n \log n}\\
&\geq s + \frac{2\sqrt{n \log n}}{3} - \epsilon\sqrt{n \log n}\\
&>s
\end{align*}
\end{proof}

\begin{claim}\label{claim:q_preserved}
Let $\x \in \mathcal{C}$ be any configuration such that $\vert s \vert \ge \frac{2}{3}n$ and $b \ge 2\sqrt{n \log n}$. Then it holds that $Q > \sqrt{n \log n}$, w.h.p.
\end{claim}

\begin{proof}
W.l.o.g. we assume that $a>b$. The number of \colB-colored nodes is at least $2\sqrt{n \log n}$ and each node has probability at least $2/3$ to pick an \colA-colored node. Thus $\Expec{}{Q} > \frac{4}{3}\sqrt{n \log n}$ and we get the claim by a simple application of the additive form of the Chernoff bound.
\end{proof}

The three above claims imply that, after $\bigO(\log n)$ rounds, the process reaches a configuration such that $s \geq \frac{2}{3}n$, $q \geq \sqrt{n \log n}$ and $b \le 2 \sqrt{n \log n}$. The next claim shows that starting from any such configuration the number of undecided nodes decreases exponentially fast. Next, we show that if the process reaches a configuration such that $q \leq 12 \sqrt{n \log n}$ and $b \leq 2\sqrt{n \log n}$ then within few rounds the \threestateprocess\ converges to the configuration where all nodes support \colA.

\begin{claim}\label{claim:q_decrease}
Let $\x \in \mathcal{C}$ be any configuration such that $12 \sqrt{n \log n} \leq q \leq \frac{1}{3}n$  and $b \leq 2\sqrt{n \log n}$. Then it holds that $Q \leq q(1 - \frac{1}{9})$, w.h.p.
\end{claim}

\begin{proof}
From (\ref{exp:q}), we have:
\begin{align*}
\Expec{}{Q\vert X_t = \x} &= \frac{q^2 + 2ab}{n} \leq \frac{q^2 + 4 n \sqrt{n \log n}}{n}\\ &= q( \frac{q}{n} + \frac{4 \sqrt{n \log n}}{q})\\ 
&\leq q(\frac{1}{3} + \frac{1}{3})\\ &= q(1 - \frac{1}{3}).
\end{align*}
\noindent Thus we apply the multiplicative form of the Chernoff Bound (\ref{eq:cbmfgeq} in Appendix \ref{apx:cbmf}) with $\delta = \frac{1}{3}$
\begin{align*}
\Prob{}{Q \geq (1+\delta)\left( 1- \frac{1}{3}\right)q} &\leq e^{-q(1-\frac{1}{3})\delta^2/3}\\
&\leq e^{-\log n(1-\frac{1}{3})\delta^2/3}\\
&= \frac{1}{n^{\Theta(1)}},
\end{align*}
\noindent and thus we get that, w.h.p.,

\[Q \leq \left(1 + \frac{1}{3}\right)\left( 1- \frac{1}{3}\right)q = q\left( 1- \frac{1}{9}\right). \] 
\end{proof}

\begin{claim}\label{claim:a_wins}
Let $\gamma$ be any positive constant and let $\x \in \mathcal{C}$ be any configuration such that $q \leq \gamma \sqrt{n \log n}$ and $b \leq 2\sqrt{n \log n}$ then the \threestateprocess\ reaches a configuration $\X'$ with $a(\X')=n$ within $\bigO(\log n)$ rounds, w.h.p.
\end{claim}

\begin{proof}
We first show that in one round the number of nodes that support   color \colB\ becomes logarithmic 
and the number of undecided nodes does not increase.
\begin{align*}
\Expec{}{B \, \vert \, X_t = \x}& = b\left(\frac{b + 2 q}{n} \right)\\ &\leq 2\sqrt{n \log n}\left(\frac{2\sqrt{n \log n} + 2 \gamma \sqrt{n \log n}}{n}\right)\\ &= 4(\gamma +1) \log n.
\end{align*}
Using the multiplicative form of the Chernoff bound we  get that $B < 8(\gamma +1) \log n$, w.h.p. We now show that the number of the undecided nodes is still $\bigO(\sqrt{n \log n})$. Indeed
\begin{align*}
\Expec{}{Q \, \vert \, X_t = \x}& = \frac{q^2}{n} + \frac{2ab}{n}\\ &\leq \gamma^2 \log n + 4 \sqrt{n \log n}.
\end{align*}
Then using the additive form of the Chernoff bound we get that $Q \leq 5 \sqrt{n \log n} $ ,w.h.p. In the next round, w.h.p., no undecided node picks a node colored of \colB\ or vice versa, so we can conclude that there are no nodes that still support \colB\   (and it easy to show that there is at least one supporter of \colA\, w.h.p.). From now on, the stochastic process is equivalent to the  classic spreading process via \pull\ operations, and thus, within $\bigO(\log n)$ rounds, all the nodes will support \colA\ w.h.p.
\end{proof}

We are now ready to show the following

\begin{lemma}[Phase $H_6$]\label{lemma:convergence}
Starting from any configuration $\x \in H_6$ with $a>b$, the \threestateprocess\ ends in the monochromatic configuration where $a = n$ within $\bigO(\log n)$ rounds, w.h.p.
\end{lemma}
\begin{proof}
Let us first assume that $s(\x) \geq n - 5\sqrt{n \log n}$ and $q(\x) \leq \sqrt{n \log n}$. This implies that $b(\x) \leq 2 \sqrt{n \log n}$ and thanks to Claim \ref{claim:a_wins} we get that the process ends in the configuration such that $a = n$ within $\bigO(\log n)$ rounds. Otherwise $s(\x) \geq \frac{2}{3}n$ and $q \geq \sqrt{n \log n}$. Then, starting from $\x$, we iteratively apply Claim \ref{claim:b_decrease} together with Claim \ref{claim:s_preserved} and Claim \ref{claim:q_preserved}:   we thus get that the process reaches a configuration $\X'$ such that $s(\X') \geq \frac{2}{3}n$, $q(\X') \geq \sqrt{n \log n}$ and $b(\X') \le 2 \sqrt{n \log n}$ within $\bigO(\log n)$ rounds. Then we iteratively apply Claim \ref{claim:q_decrease} together with Claim  \ref{claim:b_decrease}\footnote{If $b < \log n$ we cannot apply Claim \ref{claim:b_decrease} in order to show that $B$ does not overtake $2 \sqrt{n \log n}$  but we can get the claim with a simple application of Markov inequality.} and Claim \ref{claim:s_preserved} in order to ensure the  process reaches a configuration $\X''$ such that $q(\X'') \leq 12 \sqrt{n \log n}$ and $b(\X'') \leq 2 \sqrt{n \log n}$ within $\bigO(\log n)$ rounds.     We can now apply Claim \ref{claim:a_wins}  and get the  process reaches the monochromatic configuration, w.h.p. Since every step of the iterations holds w.h.p. and the number of steps is $\bigO(\log n)$, we easily obtain the thesis by a simple application of the Union Bound.
\end{proof}

\paragraph*{Phases $H_5$ and $H_7$ (starters)}
We show that if the process is in a configuration where the number of the 
undecided nodes is relatively small with respect to the bias, then in the 
next round the number of the undecided nodes becomes large while the bias 
does not decrease too much, w.h.p. This essentially implies that if the 
process starts in $H_5$ then in the next round the process moves to a 
configuration belonging to $H_4$ or $H_6$ 
(Lemma~\ref{lemma:birthofundecidedsII}), while if it starts in $H_7$ then 
in the next round it moves to $H_4$ or $H_5$ or $H_6$ 
(Lemma~\ref{lemma:birthofundecidedsIII}).

\begin{claim}\label{claim:s_does_not_decrease}
Let $\gamma,\epsilon$ be any two positive constants and $\x \in \mathcal{C}$ any configuration such that $s \geq \gamma \sqrt{n \log n}$. Then it holds that $S \geq (\gamma - \epsilon)\sqrt{n \log n}$, w.h.p.
\end{claim}

\begin{proof}
Since $S = A - B$,   we need to  provide two   bounds to the values of $A$ and $B$, respectively. We use the additive form of the Chernoff bound ((\ref{eq:cbafleq}) and (\ref{eq:cbafgeq}) in Appendix \ref{apx:cbaf}) with $\lambda = \epsilon\frac{\sqrt{n \log n}}{2})$. We have

\[\Prob{}{A \leq  \Expec{}{A \vert X_t = x} - \lambda} \leq e^{-2\lambda^2/n} = e^{-\epsilon^2\log n/2} = \frac{1}{n^{\Theta(1)}},\]

\noindent and

\[\Prob{}{B \geq  \Expec{}{B \vert X_t = x} + \lambda} \leq e^{-2\lambda^2/n} = e^{-\epsilon^2\log n/2} = \frac{1}{n^{\Theta(1)}}.\]

\noindent Then it holds that, w.h.p.,
\begin{align*}
S &\geq  \Expec{}{A \vert X_t = x} - \lambda - \Expec{}{B \vert X_t = x} - \lambda\\
&= \Expec{}{A - B \vert X_t = x} - 2 \lambda\\
&= \Expec{}{S \vert X_t = x} - 2 \lambda\\
&= s(1 + \frac{q}{n}) - \epsilon\sqrt{n \log n}\\
&\geq s - \epsilon\sqrt{n \log n}\\
&\geq\gamma\sqrt{n \log n} - \epsilon\sqrt{n \log n}\\
&=(\gamma-\epsilon)\sqrt{n \log n}.
\end{align*}
\end{proof}

The above claim and Claim \ref{claim:q_lower_bound} immediately imply the following

\begin{lemma}[Phase $H_5$]\label{lemma:birthofundecidedsII}
Starting from any configuration $\x \in H_5$ with $a>b$, the \threestateprocess\ reaches a configuration $\X' \in (H_4 \cup H_6)$ with $a>b$ in one round, w.h.p.
\end{lemma}

Concerning phase $H_7$, we have
\begin{lemma}[Phase $H_7$]\label{lemma:birthofundecidedsIII}
Starting from any configuration $\x \in H_7$ with $a>b$, the 
\threestateprocess\ reaches a configuration $\X' \in (H_4 \cup H_5 \cup H_6)$ with $a>b$ in one round, w.h.p.
\end{lemma}
\begin{proof}
Note that Claim \ref{claim:s_does_not_decrease} implies that in the next round the process does not enter in $H_1$, $H_2$ or $H_3$ w.h.p. The Lemma assumption $\x \in H_7$, i.e. $s \leq n - 5 \sqrt{n \log n}$ and $q \leq \sqrt{n \log n}$, implies that  $b \geq 2 \sqrt{n \log n}$ and thus we can apply Claim \ref{claim:q_preserved} and get that the process leaves $H_7$ because  of the growth of the undecided nodes.
\end{proof}


\section{On the tightness of Theorem~\ref{thm:mainmajority}}\label{proof:tightness}
In this section we sketch a proof of the almost-tightness result stated in 
Theorem~\ref{thm:mainmajority}.

\begin{claim}\label{claim:apxminbias} An initial configuration exists with $\abs{s} = \Theta(\sqrt{n})$ such that the process converges to the minority color with constant probability.
\end{claim}

\begin{proof}
Let us consider the configuration $\x$ such that $q(\x) = n/3, a(\x) = n/3 + \sqrt{n}$ and $b(\x) = n/3 - \sqrt{n}$. We prove that in one round there is constant probability that the bias becomes zero or negative. After that, by a simple symmetry argument, we get the claim.

Like in Claim~\ref{unstable+goodbias} we define $A^q,B^q,Q^q$ as the random variables counting the nodes that were undecided in the configuration $\x$ and that, in the next round, 
get    colored with  \colA, \colB, or undecided, respectively. Similarly $A^a$ ($B^b$) counts the nodes that support  color \colA\ (\colB) in the configuration $\x$ and that, in the next round, still support the same color.

Since it is impossible that a node supporting a color  gets the other color in the next round, it holds that $A = A^q + A^a$ and $B = B^q + B^b$. Moreover, observe  that, among these random variables, only $A^q$ and $B^q$ are mutually dependent. Thus, for any positive constant $\delta$, it holds:  
\begin{align*}
&\Prob{}{S \leq 0} = \Prob{}{B \geq A}\\
&= \Prob{}{B^q + B^b \geq A^q + A^a}\\
&\geq \Prob{}{B^q \geq A^q \, , \, B^b \geq n/3 + \delta\sqrt{n}  \, , \,  A^a \leq n/3 + \delta\sqrt{n}}\\
&= \Prob{}{B^q \geq A^q} \cdot \Prob{}{B^b \geq n/3 + \delta\sqrt{n}}\\ &\,\,\,\,\cdot \Prob{}{A^a \leq n/3 + \delta\sqrt{n}}.
\end{align*}
By applying the reverse form of Chernoff bound (see~(\ref{eq:rcbgeq})), we get that $\Prob{}{B^b \geq n/3 + \delta\sqrt{n}}$ is at least some constant, whereas the fact that $\Prob{}{A^a \leq n/3 + \delta\sqrt{n}}$ is at least some constant is an immediate consequence of the additive form of Chernoff Bound (see~\eqref{eq:cbafgeq}). 
Thus we need to show that   $\Prob{}{B^q \geq A^q}$ is at least some constant, as well.
 Note that the distribution $B^q$ conditioned to the event $Q^q = k$ is a binomial distribution with parameters $(\frac{n}{3}-k,\frac{b(\x)}{a(\x)+b(\x)})$ and with expectation \[ \Expec{}{B^q \, \vert \, Q^q = k, X_t = \x} = \frac 12 \, \left(\frac{n}{3}-k\right) - \left(\frac{n}{3}-k\right) /  \left( \frac 23 \,  \sqrt n \right)\, . \] 
 Thus, setting $\alpha =  (n/3 -k) / (\frac 23 \sqrt n)$ and $\beta = 1/(\frac13 \sqrt n -1)$ we get
\allowdisplaybreaks 
\begin{align}
&\Prob{}{B^q \geq A^q} = \sum_{k=1}^{n/3} \Prob{}{B^q \geq A^q \, \vert \, Q^q = k}\Prob{}{Q^q = k}\nonumber\\
&> \sum_{k=n/4}^{n/2} \Prob{}{B^q \geq A^q \, \vert \, Q^q = k}\Prob{}{Q^q = k}\nonumber\\
&= \sum_{k=n/4}^{n/2} \Prob{}{B^q \geq (\frac{n}{3} - k)/2) \, \vert \, Q^q = k}\Prob{}{Q^q = k}\nonumber\\
&= \sum_{k=n/4}^{n/2} \Prob{}{B^q \geq  \Expec{}{B^q\, \vert \, Q^q = k} + \alpha\, \vert \, Q^q = k}\Prob{}{Q^q = k}\nonumber\\
&= \sum_{k=n/4}^{n/2} \Prob{}{B^q \geq  \Expec{}{B^q\, \vert \, Q^q = k  }  (1+\beta) \, \vert \, Q^q = k}\Prob{}{Q^q = k}\nonumber\\
&\geq \sum_{k=n/4}^{n/2} \exp\left( -9 \beta^2 \Expec{}{B^q\, \vert \, Q^q = k  } \right)\Prob{}{Q^q = k}\label{appliedreverse}\\
&= \sum_{k=n/4}^{n/2} \exp\left(\frac{-9}{(\frac 13\sqrt{n}-1)^2} \cdot \Expec{}{B^q\, \vert \, Q^q = k}\right)\Prob{}{Q^q = k}\nonumber\\
&\geq \sum_{k=n/4}^{n/2} \exp\left(\frac{-9n}{(\frac 13\sqrt{n}-1)^2} \right)\Prob{}{Q^q = k}\nonumber\\
&= \exp\left(\frac{-9n}{(\frac 13\sqrt{n}-1)^2} \right) \sum_{k=n/4}^{n/2} \Prob{}{Q^q = k}\nonumber\\
&= \Theta(1) \sum_{k=n/4}^{n/2} \Prob{}{Q^q = k}\nonumber\\
&=\Theta(1) (1 - e^{-\Theta(n)}), \label{concentrateq}
\end{align}
where in (\ref{appliedreverse}) we used the reverse form of Chernoff bound (see~(\ref{eq:rcbgeq})) and in (\ref{concentrateq}) we used that $\Expec{}{Q^q} \approx \frac{n}{3}$ and the additive form of Chernoff bound. 
\end{proof}


%% file: trunk/conclusions.tex
\section{Conclusions}
We provided a full analysis of the \threestate\ in the parallel \pull\ 
model for the binary case showing that the resulting process converges quickly,
regardless of the initial configuration. Besides giving   tight
bounds on the convergence time, our set of results well-clarifies the main 
aspects of the process evolution and the crucial role of the undecided nodes 
in each phase of this evolution.

 An interesting open question   is  that of considering the same process in the 
multi-color case and to derive bounds on the time required to break symmetry 
from balanced configurations, as well. 

Finally, we believe that our analysis can be suitably adapted in order to show that
the \threestate\ efficiently stabilizes to a valid consensus ``regime'' \footnote{According to the notion of \emph{stabilizing almost-consensus protocol}
given in~\cite{AAE07,becchetti2016stabilizing}.} even in the presence
of a dynamic adversary that can change the state of a subset of nodes of size 
$o(\sqrt n)$ provided that the initial number of colored nodes is $\Omega(\sqrt n)$.

%% file: trunk/apx-inequalities.tex
\section{Useful inequalities}
\subsection{Chernoff Bound moltiplicative form} \label{apx:cbmf}

Let $X_1, \dots, X_n$ be independent 0-1 random variables. Let $X = \sum_{i=1}^n X_i$ and $\mu \leq \Expec{}{X} \leq \mu'$. Then, for any $0<\delta<1$ the following Chernoff bounds hold:
\begin{equation}\Prob{}{X \geq (1+ \delta)\mu} \leq e^{-\mu \delta^2 /3}.\label{eq:cbmfgeq}
\end{equation}
\begin{equation}\Prob{}{X \leq (1- \delta)\mu'} \leq e^{-\mu' \delta^2 /2}.\label{eq:cbmfleq}
\end{equation}
\subsection{Chernoff Bound additive form} \label{apx:cbaf}

Let $X_1, \dots, X_n$ be independent 0-1 random variables. Let $X = \sum_{i=1}^n X_i$ and $\mu = \Expec{}{X}$. Then the following Chernoff bounds hold:

\noindent for any $0<\lambda<n-\mu$,
\begin{equation}\Prob{}{X \leq \mu - \lambda} \leq e^{-2\lambda^2/n},\label{eq:cbafleq}
\end{equation}
for any $0<\lambda<\mu$,
\begin{equation}\Prob{}{X \geq \mu + \lambda} \leq e^{-2\lambda^2/n}.\label{eq:cbafgeq}
\end{equation}
\subsection{Reverse Chernoff Bound} \label{apx:rcb}

Let $X_1, \dots, X_n$ be independent 0-1 random variables, $X = \sum_{i=1}^n X_i$, $\mu = \Expec{}{X}$ and $\delta \in (0,1/2]$. If  $\mu \leq \frac{1}{2}n$ and $\delta^2 \mu \geq 3$ then the following bounds hold:
\begin{equation}\Prob{}{X \geq (1+ \delta)\mu} \geq e^{-9\delta^2\mu},\label{eq:rcbgeq}
\end{equation}
\begin{equation}\Prob{}{X \leq (1- \delta)\mu} \geq e^{-9\delta^2\mu}.\label{eq:rcbleq}
\end{equation}